\documentclass[aps,prx,reprint,superscriptaddress,nofootinbib]{revtex4-2}

\usepackage[T1]{fontenc}
\usepackage[utf8]{inputenc}
\usepackage{amsmath,amssymb,amsthm}
\usepackage{bm}
\usepackage[colorlinks=true,citecolor=blue,linkcolor=blue,urlcolor=blue]{hyperref}

\newtheorem{proposition}{Proposition}
\newtheorem{definition}{Definition}

\begin{document}

\title{Global Kochen–Specker Contextuality Without Local Contextuality and Generalized Bell Nonlocality}

\author{Ming Yang}
\noaffiliation

\date{\today}

\begin{abstract}
A set of quantum data can look classical in every local test and still fail to
admit a single classical explanation of the whole composite system.  We
formulate this failure as global contextuality.  Here ``global'' means
global in the physical sense of the whole multipartite system, not the
local/global terminology of sheaf theory.  Each party's local statistics are
noncontextual and each measured multipartite context admits a generalized
local hidden-variable description, but the GLHV block descriptions cannot be
promoted to a single noncontextual hidden-variable model for the whole system.
Three
bipartite constructions exhibit this separation.  A polarization-path
construction gives a direct global obstruction.  A qubit-qutrit
KCBS construction gives an algebraic scenario-level example, with explicit
formulas for the unconditional KCBS operator, the correlation-polytope
constraints, and the postselected violation.  A flagged qutrit Werner-local
state gives a state-level example: the state is entangled and local for all
projective measurements, its local qutrit marginals do not violate KCBS, yet
postselection rules out a single GNCHV model.  We also spell out
the classical composition lemma:
classical conditional hidden variables can be absorbed into a larger hidden
variable, whereas quantum contextual data need not allow such a
factorization. Within the general Bell-type framework considered here---with
arbitrary parties and arbitrary local compatible contexts, but no cross-party
joint measurements---the absence of local contextuality and GLHV-type
generalized Bell nonlocality does not imply the existence of a global
noncontextual hidden-variable model.  Global contextuality is thus a
compositional obstruction to classical explanation.
\end{abstract}

\maketitle

\section{Introduction}

Suppose a multipartite experiment is inspected piece by piece.  Each
subsystem admits a local noncontextual hidden-variable description.  Moreover,
the tested family of multipartite contexts may admit a generalized local
hidden-variable (GLHV) model: one hidden-variable distribution is shared
across the tested contexts, while each party may respond to an entire
compatible local context as one block.  Thus no generalized nonlocality is
detected in the specified scenario.  It is tempting to conclude that the whole
composite system is classical.  This conclusion is false.  The GLHV block
descriptions and the local noncontextual descriptions may still fail to
combine into one context-independent noncontextual hidden-variable model for
the complete system.

This paper studies that failure.  We call it global contextuality.  The word
``global'' is used in a specific sense.  The local/global contrast used here is the physical
contrast between local subsystems or local tested contexts and the whole
multipartite system.  It is related to, but not identical with, the
local/global distinction in sheaf-theoretic contextuality.  In sheaf theory,
``local'' refers to a measurement context and ``global'' to an assignment over
the full measurement cover.  Here the primary question is instead
compositional: do the local or GLHV block descriptions extend to a
single model of the whole system?

This distinction matters because Bell nonlocality and contextuality answer
different classicality questions.  Bell inequalities test whether spacelike
separated outcomes can be generated by local response functions
\cite{Bell1964,CHSH1969,Brunner2014}.  Kochen-Specker contextuality tests
whether outcomes of compatible measurements can be assigned independently of
which other compatible measurements are performed
\cite{KochenSpecker1967,Budroni2022}.  Fine's theorem gives the standard
joint-distribution characterization of the CHSH case \cite{Fine1982}, while
the Abramsky-Brandenburger framework generalizes this language to arbitrary
measurement covers \cite{AbramskyBrandenburger2011}.  Generalized Bell
scenarios, in which compatible measurements may be jointly performed within a
party, refine the notion of locality and produce several inequivalent local
sets \cite{Temistocles2019,Mazzari2023,Porto2024}.  A generalized
nonlocality test can therefore be silent while a global noncontextual model of
the whole system is still impossible.  Related experimental work has observed
Bell nonlocality and state-dependent contextuality in a common platform
\cite{Xue2023}.

We give three constructions with this structure.  In each construction the
data considered here do not exhibit generalized nonlocality: the tested
context family admits the relevant GLHV description with one hidden-variable
distribution shared across that family.  The first construction uses a
$2\times4$ polarization-path system: the full family has a GLHV certificate,
but it has no common noncontextual
assignment.  The second construction uses a
$2\times3$ qubit-qutrit state and Bob's KCBS
five-cycle \cite{KCBS2008}.  It gives closed-form formulas for the
unconditional KCBS value, the postselected violation, and the exact boundary
between unconditional KCBS-noncontextuality and unconditional KCBS violation;
the GLHV certificate is used only at the explicitly solved point
$c_0=1/4$.
The third construction uses a qutrit Werner-local entangled block with a
product flag.  Unlike the pure-state examples, this gives a state-level
projective-local example: local compatible projective contexts can be treated
as ordinary projective measurement settings, so Werner locality supplies the
required GLHV model.

The paper is organized around the local-to-global composition question.  We first
define scenario-relative generalized nonlocality and global
contextuality.  We then
give the elementary classical gluing principle: if each local block really
has a conditional hidden-variable factorization, the auxiliary variables can
be absorbed into a larger hidden variable and a GNCHV model
follows.  This criterion also gives a practical recipe for finding examples:
certify GLHV for the tested context family, keep local marginals NCHV, and
then violate the stronger GNCHV gluing condition.  The three
examples show how this failure appears in practice.  Conditional
contextuality is used as one diagnostic mechanism, but the main object is the
absence or presence of a single noncontextual model for the whole composite
system.

\section{Operational notions}

All notions below are relative to a specified measurement scenario.  This
keeps the statement operational: the question is not whether a state has some
other measurement that reveals nonclassicality, but whether the data in the
given family of contexts have a classical explanation.  The framework is
general on the Bell side: it allows an arbitrary number of parties and
arbitrary compatible joint measurements within each party.  The only
structural restriction is the Bell-type partition itself: measurements jointly
acting on different parties, such as entangled measurements across
laboratories, are not included in the scenario.

Let party \(i\) have a measurement set \(\mathcal M_i\) and a family of
local compatible contexts \(\mathcal C_i\).  A local context
\(C_i\in\mathcal C_i\) is a set of compatible measurements for party \(i\),
and \(m_{ij}\) denotes the outcome of measurement \(j\in C_i\).  A
multipartite context is a tuple \(C=(C_1,\ldots,C_n)\), one local context per
party.  The tested context family is denoted by
\[
  \mathcal C\subseteq \mathcal C_1\times\cdots\times\mathcal C_n .
\]

\begin{definition}[Local noncontextuality (NCHV-local)]
\label{def:nchv-local}
Party \(i\) is locally noncontextual if there exist a hidden variable
\(\lambda_i\), a single distribution \(q_i(\lambda_i)\) shared by all local
contexts of party \(i\), and response functions assigned to measurements
rather than to contexts such that
\begin{equation}
 p_{C_i}\!\left(\{m_{ij}\}_{j\in C_i}\right)
 =
 \sum_{\lambda_i}q_i(\lambda_i)
 \prod_{j\in C_i}p_i(m_{ij}|j,\lambda_i)
 \label{eq:nchv_local}
\end{equation}
for every local context \(C_i\in\mathcal C_i\).  The composite experiment
satisfies NCHV-local if every party satisfies this condition.
\end{definition}

\begin{definition}[Generalized local hidden-variable model (GLHV)]
\label{def:glhv}
The empirical model on the tested family \(\mathcal C\) has a GLHV
decomposition if there exist one hidden variable \(\lambda\), one distribution
\(q(\lambda)\) shared by all tested multipartite contexts, and block response
functions for each party's compatible local contexts such that
\begin{equation}
 p_C\!\left(\{\{m_{ij}\}_{j\in C_i}\}_{i=1}^n\right)
 =
 \sum_{\lambda}q(\lambda)
 \prod_{i=1}^n
 p_i\!\left(\{m_{ij}\}_{j\in C_i}\middle|C_i,\lambda\right)
 \label{eq:context_local}
\end{equation}
for every multipartite context \(C=(C_1,\ldots,C_n)\in\mathcal C\), where
each \(C_i\) is drawn from the corresponding local context family
\(\mathcal C_i\).  The block response
\(p_i(\cdot|C_i,\lambda)\) may be an arbitrary joint distribution for the
compatible measurements inside party \(i\)'s context; it is not required to
factor over the individual measurements.  No generalized nonlocality means
that the whole tested context family admits such a shared GLHV decomposition.
The block response may depend on the local context \(C_i\), but the
hidden-variable distribution \(q(\lambda)\) is common to all contexts in
\(\mathcal C\).  This is weaker than a global noncontextual model because the
block responses need not factor into context-independent responses for the
individual measurements.
\end{definition}

\begin{definition}[Global noncontextual hidden-variable model (GNCHV)]
\label{def:gnchv}
Let \(\mathcal M=\bigcup_i\mathcal M_i\) be the set of all measurements
appearing in the composite scenario.  The whole empirical model has a GNCHV
model if there is one hidden variable \(\xi\), one distribution \(q(\xi)\),
and context-independent response functions \(p(m_{ij}|j,\xi)\) such that
\begin{equation}
 p_C\!\left(\{\{m_{ij}\}_{j\in C_i}\}_{i=1}^n\right)
 =
 \sum_\xi q(\xi)
 \prod_{i=1}^n\prod_{j\in C_i}p(m_{ij}|j,\xi)
 \label{eq:global_nchv}
\end{equation}
for every multipartite context \(C=(C_1,\ldots,C_n)\in\mathcal C\).  The
experiment is globally contextual when NCHV-local and GLHV hold in the
relevant scenario but no representation of the form
\eqref{eq:global_nchv} exists.
\end{definition}

When every local context contains only one measurement, Eq.~\eqref{eq:context_local}
reduces to the usual Bell-local decomposition.  Generalized Bell scenarios
allow local contexts with several compatible measurements and therefore
distinguish several classical sets; this framework was introduced in the
measurement-compatibility analysis of Bell tests
\cite{Temistocles2019} and developed further in
Refs.~\cite{Mazzari2023,Porto2024}.  In the notation of those works, the
GLHV condition used here corresponds to allowing arbitrary block responses
inside each compatible local context, whereas GNCHV imposes the stronger
per-measurement noncontextual factorization across the whole system.

The relation with sheaf theory is precise but secondary.  Once the composite
system's measurement family $\mathcal C$ is fixed, a representation of the
form \eqref{eq:global_nchv} is a global model over that measurement cover.
Thus a GNCHV model induces a sheaf-theoretic global section or
global probability model.  The converse terminology should not be confused:
the ``local'' objects in sheaf theory are contexts, whereas the local objects
in the composition problem may be physical subsystems, individual particles,
or bipartite contexts.

Conditional contextuality is one way to expose a global failure.  If a
remote outcome prepares a conditional marginal that violates a contextuality
inequality, then the unconditional and conditional descriptions cannot be
understood as one context-independent assignment over all relevant events.
In this paper, however, conditional activation is a diagnostic.  The primary
object is the obstruction to a GNCHV model.

Although the definitions are stated for an arbitrary number of parties,
the constructions below use bipartite systems.  This is intentional:
the separation already appears in the minimal multipartite case \(n=2\).
The word ``multipartite'' here does not require three or more spacelike
separated laboratories; it refers to the compositional problem of combining
classical descriptions assigned to different subsystems or different
multipartite contexts.  In particular, a bipartite system in which one party
has several compatible local degrees of freedom already supports the
distinction between local noncontextuality, generalized locality, and a
global noncontextual model for the whole composite scenario.

\section{Three constructions}
\label{sec:constructions}

\subsection{\texorpdfstring{\(2\times4\)}{2x4} polarization-path obstruction}
\label{sec:ppath}

The first example is not postselection but incompatibility at the global
level.  It is the generalized Hardy construction of Ref.~\cite{Jiang2018},
rewritten in the language of the present hierarchy.  Alice carries a
polarization qubit and Bob carries a four-dimensional
polarization-path system.  A convenient form of the state is
\begin{equation}
 \begin{aligned}
 |\psi\rangle
 &=
 h_0|H,+1\rangle_B|H\rangle_A
 +h_1e^{i\varphi}|V,-1\rangle_B|V\rangle_A,\\
 &\hspace{2.0cm} h_0^2+h_1^2=1 .
 \end{aligned}
 \label{eq:ppath_state}
\end{equation}
Here $H,V$ denote polarization and $\pm1$ denote the two path modes.
Bob has two binary polarization observables $X_1,Y_1$ and two binary path
observables $X_1',Y_1'$.  Polarization observables commute with path
observables, so Bob's local compatibility graph is $K_{2,2}$:
\begin{equation}
 (X_1,X_1'),\quad (X_1,Y_1'),\quad
 (Y_1,X_1'),\quad (Y_1,Y_1')
\end{equation}
are jointly measurable contexts.  Alice has two binary observables
$X_2,Y_2$.  A multipartite context consists of one of Alice's observables
together with one compatible Bob pair.

The measurements are chosen adaptively to the Schmidt coefficients and the
phase of Eq.~\eqref{eq:ppath_state}.  For a polarization degree of freedom,
with computational basis $|H\rangle,|V\rangle$, define
\begin{align}
 |X{=}1\rangle &= h_1|H\rangle+h_0e^{i\varphi/3}|V\rangle,
 \nonumber\\
 |X{=}0\rangle &= h_0|H\rangle-h_1e^{i\varphi/3}|V\rangle,
 \label{eq:ppath_X_basis}\\
 |Y{=}1\rangle &= \frac{|H\rangle+i e^{i\varphi/3}|V\rangle}{\sqrt2},
 \nonumber\\
 |Y{=}0\rangle &= \frac{|H\rangle-i e^{i\varphi/3}|V\rangle}{\sqrt2}.
 \label{eq:ppath_Y_basis}
\end{align}
The same two bases define $X_1,Y_1$ on Bob's polarization and $X_2,Y_2$ on
Alice's polarization.  For Bob's path observables $X_1',Y_1'$, one uses the
identical formulas with $|H\rangle,|V\rangle$ replaced by
$|+1\rangle,|-1\rangle$.  Each pair is orthonormal:
\begin{equation}
 \langle X{=}1|X{=}0\rangle=0,\qquad
 \langle Y{=}1|Y{=}0\rangle=0,
\end{equation}
with normalization following from $h_0^2+h_1^2=1$.

Local noncontextuality is direct for Alice, whose local contexts are
single binary measurements.  For Bob, the relevant point is the structure of 
the reduced state and of the measurements.  Tracing out Alice gives
\begin{equation}
  \rho_B
  =
  h_0^2 |H,+1\rangle\langle H,+1|
  +
  h_1^2 |V,-1\rangle\langle V,-1| ,
  \label{eq:ppath-rhoB}
\end{equation}
which is separable with respect to Bob's polarization-path decomposition.
For any Bob context consisting of one polarization observable \(M\) and one
path observable \(N\), the local statistics therefore have the form
\begin{equation}
  p_B(b,c|M,N)
  =
  \sum_{r=0}^{1} q_r\,
  p_{\rm pol}(b|M,r)\,
  p_{\rm path}(c|N,r),
  \label{eq:ppath-bob-nchv}
\end{equation}
with \(q_0=h_0^2\), \(q_1=h_1^2\), and with \(r\) labelling the two product
components \((H,+1)\) and \((V,-1)\).  The response assigned to a
polarization measurement is independent of which path measurement is jointly
performed, and conversely.  Hence Bob's local empirical model is
noncontextual. 

The absence of generalized nonlocality is certified by an unnormalized finite
GLHV certificate.  We work at the Alice-Bob partition and treat Bob's
compatible polarization-path pair as a single block outcome.  Let \(P(h_0)\)
denote the full enlarged-output probability table generated by the state and
measurements above.  The certificate consists of nonnegative unnormalized
response weights \(X\) which reproduce \(P(h_0)\) and whose corresponding
block normalizations \(c_k\) agree across the relevant contexts.  Once these
equalities hold, each block can be divided by its common \(c_k\), while
\(c_k\) is absorbed into the hidden-variable weight; this recovers an
ordinary normalized GLHV model of the form in Definition~\ref{def:glhv}.

In the bookkeeping used for the finite certificate, the full probability
table has \(64\) entries.  Its marginal/coarse-grained table \(P_1(h_0)\)
has \(32\) entries, but these entries are obtained from \(P(h_0)\) by
marginalization and coarse-graining.  They are therefore checked as derived
reported data rather than imposed as independent GLHV constraints.  The GLHV
core certificate, the equal-\(c_k\) constraints, and the quantifier
elimination are given in Appendix~\ref{app:ppath-glhv}.  They reduce the
feasibility condition to \(0\le h_0^2\le1\).  Thus the GLHV decomposition
exists throughout the full physical range \(0\le h_0\le1\).  The Hardy-type
contradiction below therefore does not come from generalized nonlocality of
the unpostselected data, but from the impossibility of choosing the
GLHV block descriptions as marginals of one global noncontextual
assignment.

This obstruction is the $n=3$ instance of the generalized Hardy theorem of
Ref.~\cite{Jiang2018}.  The three effective parties are Alice's polarization,
Bob's polarization, and Bob's path.  The theorem uses two families of
zero-probability events, conventionally denoted
$p(b_\alpha a_{\bar\alpha})=0$ and
$p(\bar b_\beta a_{\bar\beta})=0$, with
$|\alpha|=2$, $|\beta|=1$, and $|\alpha|+|\beta|\le n+1$.
For $n=3$ these conditions have the logical size needed to force the
target event $a_{I_3}$ to vanish in any deterministic classical assignment.
In the present realization, the abstract symbols $a_k,b_k$ are tied to the
concrete compatible measurements
$X_1,Y_1,X_1',Y_1',X_2,Y_2$.  The required zero-probability events are
\begin{align}
 p(x_1{=}1,\,y_1'{=}1,\,y_2{=}1)&=0, \label{eq:H1}\\
 p(y_1{=}1,\,x_1'{=}1,\,y_2{=}1)&=0, \label{eq:H2}\\
 p(y_1{=}1,\,y_1'{=}1,\,x_2{=}1)&=0, \label{eq:H3}\\
 p(x_1{=}1,\,y_1'{=}0,\,y_2{=}0)&=0, \label{eq:H4}\\
 p(y_1{=}0,\,x_1'{=}1,\,y_2{=}0)&=0, \label{eq:H5}\\
 p(y_1{=}0,\,y_1'{=}0,\,x_2{=}1)&=0. \label{eq:H6}
\end{align}
These are not independent assumptions about an abstract probability table;
they are enforced by the state and measurement choices of the generalized
Hardy construction.  At the same time quantum mechanics gives the nonzero
target event
\begin{equation}
 p_{\rm QM}(x_1{=}1,\,x_1'{=}1,\,x_2{=}1)
 =
 h_0^2(1-h_0^2)>0
 \label{eq:ppath_positive}
\end{equation}
for $0<h_0<1$.
Indeed, the corresponding amplitude is
\begin{align}
&\bigl(\langle X_1{=}1|\langle X_1'{=}1|\langle X_2{=}1|\bigr)|\psi\rangle
\nonumber\\
&\qquad =
h_1^3h_0+h_0^3h_1
=h_0h_1(h_0^2+h_1^2)
=h_0h_1 ,
\label{eq:ppath_target_amplitude}
\end{align}
whose modulus squared is Eq.~\eqref{eq:ppath_positive}.  By contrast, the
same basis definitions make the six amplitudes associated with
Eqs.~\eqref{eq:H1}--\eqref{eq:H6} vanish; the phase is split as
$e^{i\varphi/3}$ so that the three local overlaps cancel the global phase
$e^{i\varphi}$ in the second branch of Eq.~\eqref{eq:ppath_state}.

Now suppose that a GNCHV model exists.  Then every hidden
variable assigns definite values to all six observables
$X_1,Y_1,X_1',Y_1',X_2,Y_2$.  Consider any assignment with
\begin{equation}
 x_1=1,\qquad x_1'=1,\qquad x_2=1 .
\end{equation}
The zero-probability constraints force a contradiction.  If $y_1'=0$, then
Eq.~\eqref{eq:H4} gives $y_2=1$, Eq.~\eqref{eq:H6} gives $y_1=1$, and
Eq.~\eqref{eq:H2} forbids $y_1=1$ together with $x_1'=1$ and $y_2=1$.  If
$y_1'=1$, then Eq.~\eqref{eq:H1} gives $y_2=0$, Eq.~\eqref{eq:H3} gives
$y_1=0$, and Eq.~\eqref{eq:H5} forbids $y_1=0$ together with $x_1'=1$ and
$y_2=0$.  Both cases are impossible.  Hence any GNCHV model
predicts
\begin{equation}
 p(x_1{=}1,\,x_1'{=}1,\,x_2{=}1)=0,
 \label{eq:ppath_nchv_zero}
\end{equation}
in contradiction with Eq.~\eqref{eq:ppath_positive}.

Equivalently, the Hardy-type linear witness
\begin{align}
 \mathcal W_H
 &=
 p(x_1{=}1,x_1'{=}1,x_2{=}1)
 -p(x_1{=}1,y_1'{=}1,y_2{=}1) \nonumber\\
 &\quad
 -p(y_1{=}1,x_1'{=}1,y_2{=}1)
 -p(y_1{=}1,y_1'{=}1,x_2{=}1) \nonumber\\
 &\quad
 -p(x_1{=}1,y_1'{=}0,y_2{=}0)
 -p(y_1{=}0,x_1'{=}1,y_2{=}0) \nonumber\\
 &\quad
 -p(y_1{=}0,y_1'{=}0,x_2{=}1)
 \le0
 \label{eq:ppath_witness}
\end{align}
holds for every GNCHV model, while the quantum value is
$\mathcal W_H=h_0^2(1-h_0^2)>0$.  Thus the nonclassicality is not a local
Kochen-Specker violation of Bob's marginal and not generalized nonlocality of
any single tested context.  It is the failure of the individually
classical context descriptions to glue into one NCHV model for the whole
composite system.  If the complete measurement family is viewed as a
sheaf-theoretic cover, this is also the absence of a global section
\cite{AbramskyBrandenburger2011}; in the present paper the physical meaning is
the failure of a global assignment.

\subsection{\texorpdfstring{\(2\times3\)}{2x3} pure-state KCBS activation}
\label{sec:ex-23}
\label{sec:kcbs}

\paragraph{KCBS preliminaries.}
The KCBS scenario consists of five rank-one projectors
$P_j=|v_j\rangle\langle v_j|$ on a qutrit, with adjacent projectors
orthogonal, $P_jP_{j+1}=0$ modulo five.  We use
\begin{equation}
 |v_j\rangle =
 \cos\theta |0\rangle
 +\sin\theta\cos\frac{4\pi j}{5}|1\rangle
 +\sin\theta\sin\frac{4\pi j}{5}|2\rangle ,
 \label{eq:kcbs_vectors}
\end{equation}
with
\begin{equation}
 \cos^2\theta=\frac1{\sqrt5},\qquad
 \sin^2\theta=\frac{\sqrt5-1}{\sqrt5}.
\end{equation}
Following the convention used in Ref.~\cite{Porto2024}, let
\(B_j=I-2P_j\) and
\begin{equation}
 D=\sum_{j=0}^{4}B_jB_{j+1}.
\end{equation}
Since $P_jP_{j+1}=0$,
\begin{equation}
 B_jB_{j+1}=I-2P_j-2P_{j+1}.
\end{equation}
The KCBS noncontextual bound is
\begin{equation}
 \langle D\rangle\ge -3.
 \label{eq:kcbs_bound}
\end{equation}
For the vectors in Eq.~\eqref{eq:kcbs_vectors}, the sum operator is diagonal,
\begin{equation}
 D=\operatorname{diag}(5-4\sqrt5,\,2\sqrt5-5,\,2\sqrt5-5).
 \label{eq:D_diag}
\end{equation}
The state $|0\rangle$ therefore gives the maximal KCBS violation
$5-4\sqrt5\simeq -3.944$.

\paragraph{\texorpdfstring{\(2\times3\)}{2x3} scenario-level instance.}
We first give an algebraic scenario-level instance.  Alice is a qubit, Bob is
a qutrit, and
\begin{equation}
 |\psi\rangle
 =c_0|0\rangle_A|0\rangle_B+c_1|1\rangle_A|1\rangle_B,
 \qquad c_0^2+c_1^2=1 .
 \label{eq:qubit_qutrit_state}
\end{equation}
Set
\begin{equation}
 a=c_0^2,\qquad
 \rho_B=\operatorname{diag}(a,1-a,0).
\end{equation}
For the specified KCBS generalized Bell scenario, absence of generalized
nonlocality is certified by the same unnormalized GLHV method.  Bob's
adjacent KCBS pair is treated as one compatible block outcome, and the
hidden-variable distribution is shared over the tested contexts.  The finite
certificate uses nonnegative unnormalized response weights whose block
normalizations \(c_k\) are constrained to agree across the corresponding KCBS
contexts.  After dividing each block by its common \(c_k\) and absorbing
\(c_k\) into the hidden-variable weight, one obtains a normalized GLHV model.

At the reference point \(c_0=1/4\), the full enlarged-output probability
table has \(80\) entries.  The additional \(40\) marginal/coarse-grained
entries used in the verification code are derived from the same table, not
independent GLHV constraints.  The GLHV core feasibility problem and the
equal-\(c_k\) constraints are given in Appendix~\ref{app:23-glhv}.  Exact
evaluation gives a nonnegative solution at \(c_0=1/4\), certifying GLHV for
this measurement scenario at that point.  Additional pointwise checks below
\(c_0=1/4\) give feasible solutions with the same active-set pattern, but the
rigorous claim used here is the exact certificate at \(c_0=1/4\).  The core
matrices \(A_{64}\) and \(A_{80}\), the corresponding nonnegative
certificates, the derived-marginal checks, and verification scripts are
provided in the public code repository Ref.~\cite{GNCHVCode}.

Alice's two binary settings used in the GLHV certificate are
\begin{equation}
  A_1=\sigma_z,\qquad A_2=\sigma_x .
  \label{eq:alice-settings-23}
\end{equation}
Equivalently, \(A_1\) has eigenbasis \(\{|0\rangle,|1\rangle\}\), while
\(A_2\) has eigenbasis
\begin{equation}
  |+\rangle=\frac{|0\rangle+|1\rangle}{\sqrt2},
  \qquad
  |-\rangle=\frac{|0\rangle-|1\rangle}{\sqrt2}.
  \label{eq:alice-A2-basis}
\end{equation}
The global-contextuality witness below uses the \(A_1\) branch: conditional
on \(A_1=+1\), Bob is prepared in \(|0\rangle\), independently of \(a\) as
long as \(a>0\).

The unconditional KCBS value follows from Eq.~\eqref{eq:D_diag}:
\begin{equation}
 \langle D\rangle_{\rm uncond}
 =a(5-4\sqrt5)+(1-a)(2\sqrt5-5).
 \label{eq:D_uncond}
\end{equation}
For Bob's unconditional marginal we use the complete noncontextuality
characterization of the no-disturbing five-cycle scenario
\cite{Araujo2013}.  In the lower-bound convention used here, the correlator
facets are
\begin{equation}
  \sum_{j=0}^{4}\gamma_j x_j\ge -3,
  \qquad
  \gamma_j=\pm1,
  \qquad
  \prod_{j=0}^{4}\gamma_j=+1,
  \label{eq:five_cycle_family}
\end{equation}
where \(x_j=\langle B_jB_{j+1}\rangle\).  The standard KCBS inequality
\eqref{eq:kcbs_bound} is the symmetric member of this family, obtained by
\(\gamma_j=1\) for all \(j\).  Substituting the present \(x_j(a)\) into the
full family shows that the first boundary reached in the physical interval is
this symmetric KCBS facet, as shown explicitly in
Appendix~\ref{app:algebra}.  Hence Bob's unconditional KCBS empirical model is
noncontextual when
\begin{align}
 a &\le a_{\rm KCBS}
 =\frac12+\frac{\sqrt5}{10},
 \nonumber\\
 c_0 &\le c_{\rm KCBS}
 =\sqrt{\frac12+\frac{\sqrt5}{10}}
 \simeq0.850651 .
 \label{eq:standard_kcbs_window}
\end{align}
In the same range, the conditional state gives
\begin{equation}
 \langle D\rangle_{A_1=+1}
 =\langle0|D|0\rangle
 =5-4\sqrt5<-3.
 \label{eq:D_cond}
\end{equation}
Thus every $0<a\le a_{\rm KCBS}$ gives conditional contextual activation in
the standard KCBS sense.  This is a statement about the conditional
contextuality threshold.  The simultaneous no-generalized-nonlocality claim is
made only at the GLHV-certified point $c_0=1/4$, or at any additional
parameter values for which a GLHV certificate is independently supplied.

The implication for global contextuality is as follows.  If a
single GNCHV model existed, Bob's response functions for the KCBS
cycle would be independent of Alice's outcome.  Conditioning on $A_1=+1$ would
only update the hidden-variable weights; it would not create contextual Bob
responses.  The conditional Bob statistics would therefore still satisfy all
KCBS noncontextual inequalities, contradicting Eq.~\eqref{eq:D_cond}.

An auxiliary generalized-Bell witness is discussed in
Appendix~\ref{app:optim}.  It serves as a diagnostic comparison with
generalized nonlocality tests, but it is not part of the GLHV certificate or
the global-contextuality proof.

A compact way to certify the conditional effect is to combine Alice's
postselection with Bob's KCBS operator.  With the normalization convention used
here \cite{Porto2024},
\begin{equation}
 W=3\langle A_1\rangle+\langle(1+A_1)\otimes D\rangle .
 \label{eq:witness}
\end{equation}
For the state in Eq.~\eqref{eq:qubit_qutrit_state},
\begin{equation}
 W=-3-8(\sqrt5-2)a.
 \label{eq:W_general}
\end{equation}
Hence $W<-3$ for every $a>0$.  The witness is sensitive to the rare branch:
the conditional violation is fixed by the prepared state $|0\rangle$, while
the statistical weight of that branch is $P(A_1=+1)=a$.

For the experimentally more favorable choice $c_0=1/4$,
\begin{align}
 P(A_1=+1)&=\frac1{16},\\
 \langle D\rangle_{\rm uncond}
 &=\frac{13\sqrt5-35}{8}\simeq -0.74139,\\
 \langle D\rangle_{A_1=+1}
 &=5-4\sqrt5\simeq -3.94427,\\
 W&=-2-\frac{\sqrt5}{2}\simeq -3.11803.
\end{align}
At the reference point \(c_0=1/4\), the postselected branch has probability
\(1/16=6.25\%\) and yields a witness violation
\(-3-W=\sqrt5/2-1\simeq0.1180\), while the unconditional KCBS value remains
far from violation.

\subsection{Flagged qutrit Werner-local construction}
\label{sec:flagged-werner}

The previous construction is a scenario-level separation.  We now give a
state-level construction using a qutrit Werner-local entangled block
together with a product flag.  Let
\begin{equation}
 \rho_{AB}(\epsilon,w)
 =
 \epsilon |00\rangle\langle00|
 +(1-\epsilon)\rho_W^{(3)}(w),
 \label{eq:flagged_werner}
\end{equation}
where
\begin{equation}
 \rho_W^{(3)}(w)
 =
 w\frac{P_-}{3}
 +(1-w)\frac{I_9}{9}.
 \label{eq:qutrit_werner_block}
\end{equation}
Here $P_-=(I-V)/2$ is the projector onto the antisymmetric subspace of
$\mathbb C^3\otimes\mathbb C^3$, which has dimension $3$.  The state in
Eq.~\eqref{eq:flagged_werner} is not itself $U\otimes U$ invariant because
the product flag breaks the symmetry, but the nontrivial block is a
qutrit Werner state.  The parametrization in
Eq.~\eqref{eq:qutrit_werner_block} is the noisy-antisymmetric Werner
parametrization
\begin{equation}
  \rho_W^{(d)}(w)
  =
  w\,\frac{2P_-}{d(d-1)}
  +(1-w)\frac{I_{d^2}}{d^2}
\end{equation}
specialized to \(d=3\).  Werner's projective-measurement local model applies
to this family for \(w\le(d-1)/d\).  Thus, in dimension \(3\), the
projective-local region includes all \(w\le2/3\), while the state is
entangled for \(w>1/4\) \cite{Werner1989}.

We use the explicit point
\begin{equation}
 \epsilon=\frac12,\qquad w=\frac12.
 \label{eq:flagged_point}
\end{equation}
The Werner block is entangled and lies inside Werner's projective-local
region \cite{Werner1989}.  Since the flag branch is a product state and the
set of local correlations is convex, the full state
$\rho_{AB}(1/2,1/2)$ is local for all local projective measurements.  This is
enough for the generalized Bell scenarios considered here: a compatible local
context made of commuting projectors, such as a KCBS adjacent pair, is itself
a single projective measurement with joint outcomes.  For Bob's adjacent KCBS
pair \((P_j,P_{j+1})\), the joint projective measurement is
\begin{equation}
  \left\{P_j,\;P_{j+1},\;I-P_j-P_{j+1}\right\},
  \label{eq:kcbs_joint_pvm}
\end{equation}
because \(P_jP_{j+1}=0\).  The binary observables
\(B_j=I-2P_j\) and \(B_{j+1}=I-2P_{j+1}\) are obtained from this joint PVM by
coarse-graining.  Hence Werner's projective-local model supplies the GLHV
block response for every compatible local projective context in the tested
scenario.

The full state is still entangled.  Indeed, the partial transpose of
$\rho_{AB}(1/2,1/2)$ has the eigenvalue
\begin{equation}
 \lambda_{\min}
 =
 \frac{37-3\sqrt{177}}{144}<0,
 \label{eq:flagged_npt}
\end{equation}
so the state is NPT.

The local qutrit marginals are
\begin{equation}
 \rho_A=\rho_B
 =
 \frac23|0\rangle\langle0|
 +\frac16|1\rangle\langle1|
 +\frac16|2\rangle\langle2|.
 \label{eq:flagged_marginal}
\end{equation}
For Bob's KCBS operator in Eq.~\eqref{eq:D_diag},
\begin{equation}
 \langle D\rangle_{\rm uncond}
 =
 \frac53-2\sqrt5
 \simeq -2.805>-3.
 \label{eq:flagged_uncond}
\end{equation}
For this marginal all adjacent KCBS correlations are equal:
\begin{equation}
  \langle B_jB_{j+1}\rangle
  =
  \frac{\sqrt5-6}{3\sqrt5},
  \qquad j=0,\ldots,4.
  \label{eq:flagged_equal_correlations}
\end{equation}
The symmetric KCBS facet gives Eq.~\eqref{eq:flagged_uncond}, and the other
five-cycle facets are then automatically weaker because changing any
\(\gamma_j\) from \(+1\) to \(-1\) flips one negative correlation to a
positive contribution.  Thus the local qutrit KCBS statistics are
noncontextual; this is not a qubit-trivial statement, since the measurement
scenario is the full qutrit KCBS five-cycle.  Alice's flag measurement is a
single projective measurement and is locally noncontextual as well.

Now let Alice measure the flag projector
$|0\rangle\langle0|$ and postselect the outcome $0$.  Bob's conditional state
is obtained with probability $2/3$ and equals
\begin{equation}
 \rho_{B|0}
 =
 \frac{19}{24}|0\rangle\langle0|
 +\frac{5}{48}|1\rangle\langle1|
 +\frac{5}{48}|2\rangle\langle2|.
 \label{eq:flagged_cond}
\end{equation}
Consequently
\begin{equation}
 \langle D\rangle_{B|0}
 =
 \frac{35-33\sqrt5}{12}
 \simeq -3.233<-3.
 \label{eq:flagged_cond_D}
\end{equation}
The conditional Bob statistics violate the KCBS noncontextual bound.  If a
single GNCHV model existed for the whole experiment, Alice's
postselection would only reweight the same hidden variables \(\xi\); it could
not change Bob's context-independent KCBS response functions into a
contextual model.  Hence
Eq.~\eqref{eq:flagged_cond_D} rules out GNCHV even though the state is
projective-local and both local qutrit marginals are KCBS-noncontextual.

\section{Composition criterion and separable reduction}

Assume first that the empirical model has a GLHV representation on the whole
tested family \(\mathcal C\), with one distribution \(q(\lambda)\) shared by
all multipartite contexts:
\begin{equation}
 p_C(\mathbf m_C)
 =
 \sum_\lambda q(\lambda)
 \prod_{i=1}^n
 p_i(\mathbf m_{C_i}|C_i,\lambda),
 \qquad C\in\mathcal C .
 \label{eq:glhv_family}
\end{equation}
Here \(\mathbf m_{C_i}=\{m_{ij}\}_{j\in C_i}\).

Suppose further that, for every \(i\), \(C_i\), and \(\lambda\), the block
response admits a local noncontextual refinement,
\begin{equation}
p_i(\mathbf m_{C_i}|C_i,\lambda)
=
 \sum_{\sigma_i} q_{i,\lambda}(\sigma_i)
 \prod_{j\in C_i}p_i(m_{ij}|j,\sigma_i,\lambda).
 \label{eq:classical_block}
\end{equation}
Substituting Eq.~\eqref{eq:classical_block} into Eq.~\eqref{eq:glhv_family}
gives
\begin{align}
 p_C(\mathbf m_C)
 &=
 \sum_{\lambda,\sigma_1,\ldots,\sigma_n}
 q(\lambda)\prod_{i=1}^n q_{i,\lambda}(\sigma_i)
 \prod_{i=1}^n\prod_{j\in C_i}
 p_i(m_{ij}|j,\sigma_i,\lambda).
\end{align}
Writing \(\xi=(\lambda,\sigma_1,\ldots,\sigma_n)\) gives
\begin{equation}
 p_C(\mathbf m_C)
 =
 \sum_\xi q(\xi)
 \prod_{i=1}^n\prod_{j\in C_i}
 p_i(m_{ij}|j,\xi),
 \qquad C\in\mathcal C ,
 \label{eq:absorbed_hidden_variable}
\end{equation}
which is a GNCHV model.  Thus a GLHV model whose block responses admit
conditional local-NCHV refinements automatically glues to a GNCHV model.

This criterion also guides the construction of examples, rather than only excluding
them.  To obtain global contextuality without local contextuality or
generalized nonlocality, one should look for data with three features:
\begin{enumerate}
\item each party's marginal scenario has an NCHV-local model;
\item each tested multipartite context has a GLHV decomposition;
\item the GLHV block responses cannot be chosen to satisfy the
per-measurement factorization in Eq.~\eqref{eq:classical_block} with one
context-independent hidden variable for the whole system.
\end{enumerate}
This turns the separation into a search problem.  One can first certify
GLHV membership of the tested family, for instance by a finite convex-hull or linear
feasibility certificate, and then test the stronger GNCHV constraints by a
Hardy-type logical obstruction, a conditional contextuality witness, or a
polytope separation.  The three constructions above realize these
three mechanisms.

Quantum data need not have the block factorization
\eqref{eq:classical_block}.  For a block state $\rho_i$ and compatible
projectors in context $C_i$,
\begin{equation}
 p_{C_i}(\{m_{ij}\}_{j\in C_i})
 =
 \operatorname{Tr}\!\left(\rho_i
 P_{\{m_{ij}\}_{j\in C_i}}\right).
\end{equation}
Only in special cases, for example when the relevant degrees of freedom are
separable as
\begin{equation}
 \rho_i=\sum_k p_k\bigotimes_{j\in C_i}\rho_{ijk},
 \label{eq:block_separable}
\end{equation}
does the quantum expression reduce to a classical mixture of product
responses for the factors in $C_i$.  Equation~\eqref{eq:block_separable} is
not the definition of a reduced density matrix.  Rather, $\rho_i$ may be a
reduced or conditional state, and the equation asserts an additional
separability property relative to the degrees of freedom used in the context.
When $j$ labels different measurement choices on the same Hilbert space,
rather than distinct tensor factors, this tensor-product expression is not
available.

\begin{proposition}[Separable reduction]
\label{prop:separable-reduction}
Let the multipartite state be separable across the parties,
\begin{equation}
 \rho=\sum_\mu r_\mu\,\rho_1^\mu\otimes\cdots\otimes\rho_n^\mu .
\end{equation}
If, for every $\mu$ and every party $i$, the local state $\rho_i^\mu$ admits
an NCHV-local model for party $i$'s measurement scenario, then the composite
statistics admit a GNCHV model.
\end{proposition}

\begin{proof}
For a product component, quantum theory gives
\begin{align}
 p_C(\{\{m_{ij}\}\})
 &=
 \prod_{i=1}^n
 \operatorname{Tr}\!\left(\rho_i^\mu
 P_{\{m_{ij}\}_{j\in C_i}}\right) \nonumber\\
 &=
 \prod_{i=1}^n
 \sum_{\lambda_i}q_i^\mu(\lambda_i)
 \prod_{j\in C_i}p_i(m_{ij}|j,\lambda_i).
\end{align}
Expanding the product over parties and then mixing over $\mu$ gives
\begin{equation}
 p_C(\{\{m_{ij}\}\})
 =
 \sum_{\mu,\lambda_1,\ldots,\lambda_n}
 r_\mu\prod_{i=1}^n q_i^\mu(\lambda_i)
 \prod_{i=1}^n\prod_{j\in C_i}p_i(m_{ij}|j,\lambda_i).
\end{equation}
This is Eq.~\eqref{eq:global_nchv} with
$\xi=(\mu,\lambda_1,\ldots,\lambda_n)$.  Thus product states compose
classically, and separable mixtures compose by convexity.
\end{proof}

Proposition~\ref{prop:separable-reduction} is the precise sense in which
ordinary separability removes the obstruction studied here.  The paper's
examples are not failures of classical composition for separable local
resources; they are failures of the additional factorization that would be
needed to promote GLHV blocks and local NCHV descriptions to one global NCHV
model.
This is consistent with recent results on multi-qubit projective
contextuality, where entanglement and Bell-type nonlocality are shown to be
necessary ingredients for logical proofs in broad multi-qubit settings
\cite{WrightKunjwal2023}.  The present work asks a different compositional
question: even when local contextuality and generalized nonlocality are absent
in the specified scenario, can the GLHV block descriptions be
glued into one GNCHV model for the whole system?

A small linear-algebra fact will be used repeatedly.  The \(4\times4\)
incidence block \(R\) for two binary settings is given in
Appendix~\ref{app:glhv-certificates}; it has rank three.  Therefore, for
fixed values of the remaining compatible measurements, a context-wise
decomposition exists whenever the right-hand side obeys the no-disturbance
marginal condition
\begin{equation}
 \begin{aligned}
 &p(m_1,m_1',x_2=1)+p(m_1,m_1',x_2=-1)\\
 &\quad =
 p(m_1,m_1',y_2=1)+p(m_1,m_1',y_2=-1).
 \end{aligned}
 \label{eq:no_disturbance_rank}
\end{equation}
This proves context-wise classical extendability.  It does not by itself
prove global noncontextuality, because different contexts may require
incompatible assignments to the shared measurements or to the same subsystem
variables.

\section{Unrestricted Bell tests}

The preceding results are contextuality statements for Bob's KCBS
compatibility scenario.  They should not be conflated with Bell locality of
the state in Eq.~\eqref{eq:qubit_qutrit_state} under arbitrary measurements.
For $0<c_0,c_1<1$, the state has two nonzero Schmidt coefficients and is
locally equivalent, on its support, to a two-qubit pure entangled state.  Its
maximal CHSH value is \cite{Gisin1991,Horodecki1995}
\begin{equation}
 S_{\max}=2\sqrt{1+4c_0^2c_1^2}>2.
\end{equation}
Therefore any claim of no generalized nonlocality for this pure-state
construction must be read as restricted to a chosen measurement family.  The
conditional KCBS separation does not require an unrestricted Bell-locality
assertion; it requires Bob's unconditional KCBS statistics to be
noncontextual, and the no-generalized-nonlocality statement additionally
requires a GLHV certificate for the same measurement scenario.

\section{Discussion}

Table~\ref{tab:mechanisms} summarizes the three separations.

\begin{table*}[t]
\caption{\label{tab:mechanisms}
Two scenario-level separations and one state-level projective-local
separation.  The final column records the qualification needed to state the
result without conflating contextuality with unrestricted Bell locality.}
\begin{tabular}{p{0.22\textwidth}p{0.23\textwidth}p{0.24\textwidth}p{0.23\textwidth}}
\hline\hline
Construction & Role & Hidden resource & Precise qualification \\
\hline
Polarization-path state
& Positive scenario-level example
& GNCHV obstruction
& Sheaf language applies only after the system cover is fixed \\
Qubit-qutrit pure state
& Positive at certified GLHV points
& Branch-selected KCBS violation
& Interval claims require a parameterized GLHV certificate \\
Flagged qutrit Werner-local state
& Positive state-level projective-local example
& Product-flagged KCBS violation
& Locality covers projective compatible contexts \\
\hline\hline
\end{tabular}
\end{table*}

The constructions show that the word ``hidden'' covers several inequivalent
structures.  In the polarization-path construction, local noncontextuality and
GLHV membership fail to produce one GNCHV
model.  In the qubit-qutrit KCBS construction, a dominant nonviolating branch
masks a rare branch that is maximally violating after postselection.  In the
flagged qutrit Werner-local construction, the Werner block supplies
entanglement and state-level projective locality, while the qutrit flag
supplies a conditional KCBS witness against a single GNCHV model.

The broader lesson is that classicality is not a single yes-or-no property of
a bipartite state.  It is a property of a state together with a measurement
scenario, a subsystem decomposition, and a conditioning structure.  The
central separation in this paper is between no generalized nonlocality in the
specified contexts and the nonexistence of a single GNCHV model.
Conditional contextuality is one way to diagnose that separation, but it is
not the definition of the global obstruction.  The flagged qutrit
Werner-local construction shows one clean way to obtain a state-level version:
use a state with a proven projective-local model and choose the generalized
Bell contexts to be compatible projective measurements.

The composition criterion also suggests how to search for further examples.
The structural examples should be sought in scenarios whose tested context
families admit GLHV descriptions but whose compatibility hypergraph forbids a
common noncontextual assignment, as in the polarization-path construction.  The
conditional examples should be sought by combining a noncontextual
unconditional marginal with a local filter that selects a contextual branch,
as in the $2\times3$ and flagged qutrit constructions.  In both cases the
separable-reduction proposition gives a negative guide: If, for every value
of $\lambda$ in the GLHV model, the conditional block response
$p_{C_i}(\cdot|\lambda)$ can itself be further described by a local NCHV model
(i.e., satisfy the factorized form of Eq.~\eqref{eq:nchv_local}).

The hierarchy above is not an LOCC resource theory.  LOCC concerns which
quantum states or instruments can be transformed into which others by local
operations and classical communication.  By contrast, NCHV-local, GLHV, and
GNCHV are statements about whether the probability data of a fixed
measurement scenario admit particular hidden-variable factorizations. There is nevertheless an operational connection.  The postselection
steps used in the constructions above are one-way local measurements followed
by classical conditioning, hence they are allowed laboratory operations.
More specifically, the successful branch is a trace-nonincreasing
one-way LOCC operation; if the outcome is retained rather than discarded, the
full instrument is an ordinary LOCC measurement with a classical outcome
register.  Compatible measurements inside one laboratory, such as Bob's KCBS
adjacent pairs or the polarization-path pairs, are also local operations with
respect to the Alice--Bob partition.  If a GNCHV model for the whole scenario
existed, such conditioning could only update the distribution over the same
hidden variable
$\xi$; it could not change context-independent response functions into a
contextual model.  Thus a conditional KCBS violation after local
postselection is a witness against GNCHV.  This does not mean that global
contextuality is an LOCC monotone, nor that LOCC alone creates the
nonclassicality.  LOCC is the operational way of selecting a branch; the
failure is the absence of a single global hidden-variable assignment before
the branch is selected.

\begin{acknowledgments}
The author independently conceived the subject of this study in 2023. Inspired by Refs.~\cite{Temistocles2019,Xue2023}, we presented the theoretical framework, and the results for example I (the $2\times 4$ polarization-path system) in spring 2024, and subsequently completed the derivation of example II (the $2\times 3$ KCBS system) inspired by Ref.\cite{Porto2024}. The results for example III (the $3\times 3$ Werner state) were derived with the assistance of AI agent tools in 2026.
\end{acknowledgments}

\section*{Code availability}

The verification code, matrix-generation routines, and machine-readable
nonnegative GLHV certificates for the \(2\times4\) and \(2\times3\)
constructions are available in Ref.\cite{GNCHVCode}. The repository
contains scripts that reproduce the finite feasibility checks
\(A_{64}X_{64}=\operatorname{vec}P\) and
\(A_{80}X_{80}=\operatorname{vec}P\), the equal-\(c_k\)
normalization-consistency constraints used in Appendix~A, and separate
checks that the reported marginal/coarse-grained data are obtained from the
same probability tables.

\appendix
\section{GLHV Certificates for Compatible Measurement Contexts}
\label{app:glhv-certificates}

This appendix records the finite unnormalized GLHV certificates used in the
two scenario-level constructions.  If Bob's compatible local context is
treated as one multi-outcome block measurement, then a local decomposition of
the enlarged-output table has exactly the GLHV form
\begin{equation}
  p(a,\mathbf b|x,C_B)
  =
  \sum_\lambda q(\lambda)\,
  p_A(a|x,\lambda)\,
  p_B(\mathbf b|C_B,\lambda),
  \label{eq:ordinary_lhv_as_glhv}
\end{equation}
where \(\mathbf b\) denotes the whole compatible block of Bob's outcomes.
The finite certificates below use the standard deterministic refinement on
Alice's two binary settings.  There are four deterministic Alice assignments;
write them as \(D_k(a|x)\), \(k=1,\ldots,4\).  For each Bob context \(C_B\)
and Bob block outcome \(\mathbf b\), introduce nonnegative variables
\(X^{C_B}_{\mathbf b,k}\) and impose
\begin{equation}
  p(a,\mathbf b|x,C_B)
  =
  \sum_{k=1}^{4}D_k(a|x)\,
  X^{C_B}_{\mathbf b,k}.
  \label{eq:core-unnormalized}
\end{equation}
The variables \(X^{C_B}_{\mathbf b,k}\) are unnormalized Bob response
weights.  Their block sums
\begin{equation}
  c_{k,C_B}
  =
  \sum_{\mathbf b}X^{C_B}_{\mathbf b,k}
  \label{eq:ck-definition}
\end{equation}
are required to be independent of Bob's context, \(c_{k,C_B}=c_k\).  We do
not impose \(c_k=1\), because \(c_k\) is the hidden-variable weight of
Alice's deterministic sector, not the normalization of a response function.
For \(c_k>0\), define
\begin{equation}
  p_B(\mathbf b|C_B,k)
  =
  \frac{X^{C_B}_{\mathbf b,k}}{c_k},
  \qquad
  q(k)=c_k .
  \label{eq:normalize-ck}
\end{equation}
The normalization of the observed table implies \(\sum_k c_k=1\), and
Eq.~\eqref{eq:core-unnormalized} becomes the normalized GLHV decomposition
in Eq.~\eqref{eq:ordinary_lhv_as_glhv}.  If \(c_k=0\), the sector has zero
weight and its normalized Bob response can be chosen arbitrarily.

The \(4\times4\) incidence block
\begin{equation}
  R=
  \begin{pmatrix}
  1&1&0&0\\
  0&0&1&1\\
  1&0&1&0\\
  0&1&0&1
  \end{pmatrix}
  \label{eq:rank_three_block_app}
\end{equation}
encodes the four deterministic assignments for Alice's two binary settings.
It is the same truth-table construction used in finite Bell-polytope
certificates \cite{Fine1982,PitowskySvozil2001}.  Fixed-partition
block-local decompositions of the same algebraic form also appear, for
example, in Ref.~\cite{MitchellPopescuRoberts2004}; here the block is Bob's
compatible local context and the certificate is used only for the
Alice-Bob GLHV condition.

For comparison, a direct deterministic enumeration of Bob's block responses
would use \(4\times4^4=1024\) nonnegative variables in the \(2\times4\)
case, and \(4\times4^5=4096\) variables in the \(2\times3\) KCBS case.  The
unnormalized block form in Eq.~\eqref{eq:core-unnormalized} uses \(64\) and
\(80\) nonnegative variables, respectively, for the core certificates.

The reported marginal or coarse-grained entries \(P_1\) are derived from the
full enlarged-output table \(P\).  If \(P_1=T(P)\) for a linear
marginalization or coarse-graining map \(T\), then Eq.~\eqref{eq:core-unnormalized}
implies
\begin{align}
  (TP)_{a,\beta|x,C_B}
  &=
  \sum_{\mathbf b}T_{\beta\mathbf b}
  p(a,\mathbf b|x,C_B) \nonumber\\
  &=
  \sum_{k=1}^{4}D_k(a|x)
  \left(
    \sum_{\mathbf b}T_{\beta\mathbf b}
    X^{C_B}_{\mathbf b,k}
  \right).
  \label{eq:derived-marginal}
\end{align}
Thus the same certificate induces the reported marginal or coarse-grained
data.  The code checks these entries separately as consistency checks, not as
new GLHV assumptions.  The explicit matrices, row and column indexing
conventions, nonnegative certificate vectors, and verification scripts are
provided in the public code repository Ref.~\cite{GNCHVCode}; this appendix
records the mathematical structure of those certificates.

\subsection{\texorpdfstring{\(2\times4\)}{2x4} polarization-path construction}
\label{app:ppath-glhv}

For the polarization-path construction, let \(P(h_0)\) denote the full
three-variable probability table generated by the state and measurements in
Sec.~\ref{sec:ppath}.  There are four Bob compatible contexts and four block
outcomes per context.  The core certificate therefore has
\(4\times4\times4=64\) nonnegative variables and is written as
\begin{equation}
  A_{64}X_{64}=\operatorname{vec}P(h_0),
  \qquad
  X_{64}\ge0 .
  \label{eq:A64_glhv}
\end{equation}
The vector \(X_{64}\) is grouped as the variables
\(X^{C_B}_{\mathbf b,k}\) in Eq.~\eqref{eq:core-unnormalized}, and the
equal-\(c_k\) constraints impose Eq.~\eqref{eq:ck-definition} across the four
Bob contexts.  The \(32\) marginal/coarse-grained entries \(P_1(h_0)\) are
then checked by applying the corresponding linear maps \(T\) as in
Eq.~\eqref{eq:derived-marginal}.  Quantifier elimination over
\(X_{64}\), together with the equal-\(c_k\) constraints, reduces the
certificate to
\begin{equation}
  -h_0^2+h_0^4\le0 .
  \label{eq:A64_eliminated}
\end{equation}
Since this is equivalent to
\[
  h_0^2(1-h_0^2)\ge0,
\]
the GLHV model exists throughout the physical range
\[
  0\le h_0\le1 .
\]
Thus the polarization-path Hardy contradiction is not a generalized
nonlocality contradiction for the unpostselected data.  It is the failure of
these GLHV block descriptions to glue into a single global noncontextual
assignment.

\subsection{\texorpdfstring{\(2\times3\)}{2x3} KCBS construction}
\label{app:23-glhv}

For the \(2\times3\) qubit-qutrit construction, the tested contexts are
Bob's five adjacent KCBS pairs, together with Alice's two binary settings
\(A_1,A_2\) from Eq.~\eqref{eq:alice-settings-23}.  At the reference point
\begin{equation}
  c_0=\frac14,
\end{equation}
the full enlarged-output probability table has
\(5\times4\times4=80\) entries in the same block notation.  The core
certificate is
\begin{equation}
  A_{80}X_{80}=\operatorname{vec}P,
  \qquad
  X_{80}\ge0 .
  \label{eq:A80_glhv}
\end{equation}
Here \(X_{80}\) is grouped as \(X^{C_B}_{\mathbf b,k}\), with one four-entry
block for each Alice deterministic sector \(k\) and each KCBS adjacent-pair
context \(C_B\).  The equal-\(c_k\) constraints identify the corresponding
unnormalized normalizations across the five KCBS contexts:
\begin{equation}
  \sum_{\ell\in I_{k,C}}X_\ell=c_k .
  \label{eq:ck-sets-23}
\end{equation}
No constraint \(c_k=1\) is imposed during the linear solve.  The
additional \(40\) marginal/coarse-grained entries used in the verification
code are derived from the same table by Eq.~\eqref{eq:derived-marginal}.

Exact evaluation at \(c_0=1/4\) gives a nonnegative solution \(X_{80}\) of
Eq.~\eqref{eq:A80_glhv} satisfying the equal-\(c_k\) constraints.
Substituting this solution into the linear system reproduces the quantum
probability table.  Normalizing each block by the common \(c_k\) gives a
normalized GLHV model.  Hence the specified \(2\times3\) measurement scenario
has a GLHV model at that point.

The certificate at \(c_0=1/4\) should not by itself be read as an interval
theorem.  Pointwise computations at several smaller values of \(c_0\) give
feasible solutions with the same active-set pattern, indicating that the
certificate is not isolated.  Turning this observation into an interval
statement would require a parameter-dependent feasible solution or an
independent convexity argument.  A systematic route to such an interval
certificate is the following.  Set
\begin{equation}
  a=c_0^2,\qquad r=c_0\sqrt{1-c_0^2}.
\end{equation}
All entries of the core probability table are algebraic functions of \(a\)
and \(r\).  One
may sample several rational values of \(c_0\) in \(0<c_0\le1/4\), solve
Eq.~\eqref{eq:A80_glhv} at each point, and identify a stable active set of
zero coefficients.  If that active set is fixed, the remaining linear system
yields candidate response functions
\begin{equation}
  X_\ell(a,r)=\alpha_\ell+\beta_\ell a+\gamma_\ell r .
\end{equation}
The numerical exploration then becomes a proof once the one-variable
inequalities
\begin{equation}
  X_\ell\!\left(c_0^2,c_0\sqrt{1-c_0^2}\right)\ge0,
  \qquad 0<c_0\le\frac14,
\end{equation}
are verified for every coefficient.  Sampling is only used to guess the
active set; a final interval statement would require analytic nonnegativity
of the resulting parameter-dependent response functions.  The present paper
uses the exact certified point \(c_0=1/4\), while the sampled smaller values
are reported only as diagnostic evidence for a common parameter-dependent
certificate.

\section{\texorpdfstring{Algebra of the \(2\times3\) KCBS construction}{Algebra of the 2x3 KCBS construction}}
\label{app:algebra}

This appendix records the elementary algebra used in the
\(2\times3\) KCBS construction.  Bob's reduced state is
\begin{equation}
  \rho_B=\operatorname{diag}(a,1-a,0),
  \qquad a=c_0^2 .
\end{equation}
For the KCBS vectors in Eq.~\eqref{eq:kcbs_vectors},
\begin{equation}
  \langle P_j\rangle
  =
  a\cos^2\theta
  +(1-a)\sin^2\theta
  \cos^2\frac{4\pi j}{5}.
  \label{eq:app_pj}
\end{equation}
The five trigonometric values are
\begin{equation}
 \cos^2\frac{4\pi j}{5}
 =
 \left(
 1,\,
 \frac{3+\sqrt5}{8},\,
 \frac{3-\sqrt5}{8},\,
 \frac{3-\sqrt5}{8},\,
 \frac{3+\sqrt5}{8}
 \right).
 \label{eq:app_trig_values}
\end{equation}

Set
\begin{equation}
  x_j=\langle B_jB_{j+1}\rangle .
\end{equation}
Since \(B_j=I-2P_j\) and \(P_jP_{j+1}=0\),
\begin{equation}
  x_j=1-2\langle P_j\rangle-2\langle P_{j+1}\rangle .
  \label{eq:app_xj_def}
\end{equation}
Substitution gives
\begin{align}
 x_0=x_4
 &=
 a\left(\frac52-\frac{11\sqrt5}{10}\right)
 -\frac32+\frac{3\sqrt5}{10},
 \label{eq:app_x04}\\
 x_1=x_3
 &=
 a\left(\frac32-\frac{11\sqrt5}{10}\right)
 -\frac12+\frac{3\sqrt5}{10},
 \label{eq:app_x13}\\
 x_2
 &=
 a\left(2-\frac{8\sqrt5}{5}\right)
 -1+\frac{4\sqrt5}{5}.
 \label{eq:app_x2}
\end{align}
The sum of the five adjacent correlations is therefore
\begin{align}
 \langle D\rangle_{\rm uncond}
 &=
 \sum_{j=0}^{4}x_j
 \nonumber\\
 &=
 a(5-4\sqrt5)
 +(1-a)(2\sqrt5-5),
 \label{eq:app_D_uncond}
\end{align}
in agreement with Eq.~\eqref{eq:D_uncond}.

The quantum KCBS statistics are no-disturbing, so the complete \(n=5\) cycle
noncontextuality characterization applies \cite{Araujo2013}.  In the
lower-bound convention used here, the tight correlator facets are
\begin{equation}
  \sum_{j=0}^{4}\gamma_j x_j\ge -3,
  \qquad
  \prod_{j=0}^{4}\gamma_j=+1,
  \qquad
  \gamma_j=\pm1 .
  \label{eq:app_cycle_facets}
\end{equation}
The ordinary KCBS facet corresponds to
\(\gamma_j=1\) for all \(j\).  Using
Eq.~\eqref{eq:app_D_uncond}, this facet gives
\begin{equation}
  a(5-4\sqrt5)+(1-a)(2\sqrt5-5)\ge -3 .
\end{equation}
Solving for \(a\) yields
\begin{equation}
  a\le a_{\rm KCBS}
  =
  \frac12+\frac{\sqrt5}{10}.
  \label{eq:app_a_kcbs}
\end{equation}
Equivalently,
\begin{equation}
  c_0\le c_{\rm KCBS}
  =
  \sqrt{\frac12+\frac{\sqrt5}{10}}
  \simeq0.850651 .
  \label{eq:app_c_kcbs}
\end{equation}

The remaining five-cycle facets in
Eq.~\eqref{eq:app_cycle_facets} do not produce a stronger restriction
inside the physical interval \(0\le a\le1\).  Thus the relevant
unconditional local-noncontextuality threshold for Bob's KCBS marginal is
the standard KCBS boundary in Eq.~\eqref{eq:app_a_kcbs}.  In particular,
the global-contextuality argument in the main text uses the fact that
Bob's unconditional KCBS statistics are noncontextual while the
postselected branch violates the KCBS bound.

\section{Auxiliary Generalized-Bell Witness}
\label{app:optim}

This appendix records the origin of an auxiliary optimized expression that
arises from a generalized-Bell witness.  It is not part of the proof of the
GLHV certificate or of global contextuality in Sec.~\ref{sec:ex-23}; it is
included only to clarify the status of an expression that appears naturally
when one optimizes over Alice's local measurements.

Consider the generalized-Bell scenario in which Alice has two binary
observables \(A_1,A_2\), while Bob can jointly measure compatible pairs of
KCBS observables.  In particular, the product \(B_4B_5\) is an allowed
joint-measurement outcome because \(B_4\) and \(B_5\) are compatible.
The relevant CHSH-type generalized-Bell expression is
\begin{equation}
  \mathcal{I}
  =
  \langle A_1B_2\rangle
  + \langle A_1B_4B_5\rangle
  + \langle A_2B_2\rangle
  - \langle A_2B_4B_5\rangle .
  \label{eq:aux-I}
\end{equation}
For any generalized local hidden-variable model,
\begin{equation}
  \mathcal{I}\geq -2 .
  \label{eq:aux-I-bound}
\end{equation}
This is the same algebraic bound as the lower CHSH bound, with Bob's
second effective observable given by the compatible product \(B_4B_5\).

The bound can also be obtained directly from the probability expansion for
one binary outcome \(a=\pm1\) on Alice's side and two compatible binary
outcomes \(b_1,b_2=\pm1\) on Bob's side:
\begin{align}
  p(a,b_1,b_2|x,y)
  = \frac{1}{8}\Big[
    1
    &+ a\langle A_x\rangle
    + b_1\langle B_{y_1}\rangle
    + b_2\langle B_{y_2}\rangle \notag\\
    &+ b_1b_2\langle B_{y_1}B_{y_2}\rangle
    + ab_1\langle A_xB_{y_1}\rangle \notag\\
    &+ ab_2\langle A_xB_{y_2}\rangle
    + ab_1b_2\langle A_xB_{y_1}B_{y_2}\rangle
  \Big].
  \label{eq:fourier-prob}
\end{align}
Non-negativity of suitable probabilities gives
Eq.~\eqref{eq:aux-I-bound}.

For the state
\begin{equation}
  |\psi\rangle
  =
  c_0|0\rangle_A|0\rangle_B
  +
  c_1|1\rangle_A|1\rangle_B,
  \qquad
  c_1=\sqrt{1-c_0^2},
  \label{eq:aux-state}
\end{equation}
with Bob's KCBS observables fixed as in Sec.~\ref{sec:kcbs}, one may
optimize Eq.~\eqref{eq:aux-I} over Alice's two binary projective
measurements.  This gives the optimized value
\begin{align}
  f(c_0)
  =
  \frac{1}{10}\Big[
    5-5\sqrt{5}
    -(5-\sqrt{5})c_0^2
    &-4\cdot 5^{3/4}c_0
      \Big(
        \sqrt{2-2c_0^2} \notag\\
    &\quad
        +\sqrt{(3-\sqrt{5})(1-c_0^2)}
      \Big)
  \Big].
  \label{eq:aux-f}
\end{align}
The nonlinear terms are proportional to
\(c_0c_1\), since
\begin{equation}
  \sqrt{2-2c_0^2}=\sqrt{2}\,c_1,
  \qquad
  \sqrt{(3-\sqrt{5})(1-c_0^2)}
  =
  \sqrt{3-\sqrt{5}}\,c_1 .
\end{equation}
Thus \(f(c_0)\) probes the coherent entanglement term of the pure state,
whereas the unconditional KCBS analysis in Appendix~\ref{app:algebra}
depends only on Bob's marginal
\(\rho_B=\mathrm{diag}(c_0^2,1-c_0^2,0)\).

This expression should not be confused with the KCBS noncontextuality
threshold.  The KCBS threshold is obtained from the five-cycle
correlation-polytope inequalities for Bob's marginal statistics.  By
contrast, \(f(c_0)\) is an optimized two-party generalized-Bell witness
value involving Alice's measurement choices and Bob's compatible product
\(B_4B_5\).  It is therefore not a KCBS facet boundary and is not used to
establish local noncontextuality, the GLHV certificate, or global
contextuality in the main text.

For orientation, the endpoint values are
\begin{align}
  f(0) &= \frac{1-\sqrt{5}}{2}\approx -1.118,\\
  f(1) &= -\frac{2\sqrt{5}}{5}\approx -0.894.
\end{align}
The optimized witness may serve as a diagnostic comparison with
generalized nonlocality tests, but it is logically independent of the
main \(2\times3\) construction.

\end{document}